\documentclass{llncs}

\usepackage{mathtools,MnSymbol1}

\input xy
\xyoption{all}

\usepackage{tikz}

\tikzstyle{vertex}=[circle, draw, inner sep=0pt, minimum size=3pt]
\newcommand{\vertex}{\node[vertex]}

\newcommand{\srec}{\!\;_{m}{\ostar}_{m'}\!\;}
\newcommand{\psrec}[2]{\;_{#1}{\ostar}_{#2}\;}

\usetikzlibrary{decorations.markings}
\begin{document}
\frontmatter          
\pagestyle{headings}  

\mainmatter              

\title{Word-Representability of Graphs\\ with respect to Split Recomposition}

\titlerunning{Title}  
%

\author{Tithi Dwary \and 
K. V. Krishna} 

\authorrunning{Tithi Dwary \and K. V. Krishna} 

\institute{Institute of Technology Guwahati, Guwahati, India,\\
	\email{tithi.dwary@iitg.ac.in};\;\;\; 
	\email{kvk@iitg.ac.in}}

\maketitle              

\begin{abstract}
In this work, we show that the class of word-representable graphs is closed under split recomposition and determine the representation number of the graph obtained by recomposing two word-representable graphs. Accordingly, we show that the class of parity graphs is word-representable. Further, we obtain a characteristic property by which the recomposition of comparability graphs is a comparability graph. Consequently, we also establish the permutation-representation number (\textit{prn}) of the resulting comparability graph. We also introduce a subclass of comparability graphs, called \textit{prn}-irreducible graphs. We provide a criterion such that the split recomposition of two \textit{prn}-irreducible graphs is a comparability graph and determine the \textit{prn} of the resultant graph.
\end{abstract}


\section{Introduction}

A simple graph is called a word-representable graph if there is a word over its vertices such that any two vertices are adjacent in the graph if and only if they alternate in the word.  The class of word-representable graphs is an interesting class not only from graph theoretical point of view but also from algebraic and computational view points. Sergey Kitaev and Steven Seif first studied this class of graphs in the context of  Perkin semigroups \cite{perkinsemigroup}. The monograph \cite{words&graphs} introduces word-representable graphs, surveys seminal contributions, and provides their connections to other contexts. The class of graphs which admit transitive orientations, known as comparability graphs, is an important subclass of word-representable graphs. This class has a special feature that they can be represented by a concatenation of permutations on their vertices, called permutationally representable graphs. In fact, the class of comparability graphs is precisely the class of permutationally representable graphs \cite{perkinsemigroup}. Every comparability graph induces a poset based on one of its transitive orientations.

A word-representable graph is said to be $k$-word-representable if it is represented by a word in which every letter appears exactly $k$ times. The smallest $k$ such that a graph is $k$-word-representable is said to be the representation number of the graph. Similarly, the smallest number $k$ such that a comparability graph is represented by a concatenation of $k$ permutations on its vertices is called the permutation-representation number (\textit{prn}) of the comparability graph. The \textit{prn} of a comparability graph is precisely the dimension of the induced poset \cite{khyodeno2}. The class of graphs with representation number at most two is characterized as the class of circle graphs (\cite{MR2914710}) and that with \textit{prn} at most two is the class of permutation graphs (cf. \cite{khyodeno1}). While there are no characterizations for the class of graphs with any of these numbers more than two, it is known that the general problems of determining the representation number of a word-representable graph, and the \textit{prn} of a comparability graph are computationally hard \cite{MR2914710,yannakakis1982complexity}.

The concepts of split recomposition,  and its inverse operation, viz., split decomposition,  were introduced by Cunningham in \cite{Cunningham_2,Cunningham_1}. The split decomposition can be seen as a generalization of modular decomposition introduced by Gallai in \cite{Gallaipaper}. The literature has several variants of split decomposition and recomposition, which include $1$-join composition \cite{1joindecomp}. In this paper, we use the version defined by Bouchet \cite{Bouchet_1}. These concepts have a large range of applications including NP-hard optimization \cite{graphcoloring,NPcompleteproblems} and the recognition of certain classes of graphs such as distance-hereditary graphs \cite{DHgraph1,DHgraph2}, circle graphs \cite{circlegraph3,circlegraph2,circlegraph1}, and parity graphs \cite{cicerone1999extension,paritygraph2}. The class of distance-hereditary graphs is characterized in terms of split components \cite{Bouchet_1,Distancehereditary2,Distancehereditary1}. Further, it was shown in \cite{Bouchet_1} that the distance-hereditary graphs are $2$-word-representable graphs.  In \cite{circlegraph3}, through split decomposition and recomposition, 2-uniform words were constructed for recognizing circle graphs. In \cite{enumeration}, split decomposition is used to give a full enumeration of some graph classes viz., plotemic, block, and various cactus graphs, and also to provide forbidden subgraph characterizations of certain graph classes.  The class of perfect graphs\footnote{A graph $G$ is said to be a perfect graph if for every induced subgraph $H$ of $G$, the chromatic number of $H$ is equal to the size of a largest clique in $H$.} has both word-representable graphs as well as non-word-representable graphs \cite{words&graphs}. However, there is no characterization available for word-representable perfect graphs. It was shown in \cite{comp_perfectgraphs} that the class of perfect graphs is closed under split recomposition.

In this paper, our focus is on exploring the word-representability of graphs with respect to split recomposition. In Section 3, we show that the class of word-representable graphs is closed under split recomposition and also determine the representation number of the resulting graph obtained by recomposing two word-representable graphs. Accordingly, we establish the word-representability of a subclass of perfect graphs, known as parity graphs. The class of parity graphs includes both bipartite graphs and distance-hereditary graphs. While the class of comparability graphs is not closed under split recomposition, in Section 4, we obtain a characteristic property by which the recomposition of comparability graphs is a comparability graph. Consequently, we establish the \textit{prn} of the resulting comparability graph. In this connection, we also investigate the \textit{prn} of a graph obtained by connecting two graphs by an edge between specified vertices. In Section 5, we introduce a subclass of comparability graphs, which we call \textit{prn}-irreducible graphs, and study their split recomposition. We provide a criteria such that the split recomposition of two \textit{prn}-irreducible graphs is a comparability graph and accordingly determine the \textit{prn} of the resultant graph.  We also observe that the split recomposition of \textit{prn}-irreducible graphs is not \textit{prn}-irreducible.

\section{Preliminaries}

This section provides the necessary background material on word-representable graphs, posets and split decomposition. For more details, one may refer to \cite{Cunningham_2,words&graphs,Trotterbook}.

A word over a finite set of letters is a finite sequence written by juxtaposing the letters of the sequence. A subword $u$ of a word $w$ is a subsequence of the sequence $w$, denoted by $u \preceq w$. For instance, $baabb \preceq aabbaaabb$. Let $w$ be a word over $A$ and $B$ be a subset of $A$. We write $w_B$ to denote the subword of $w$ that precisely consists of all occurrences of the letters of $B$. For example, if $w=baabdca$, then $w_{\{a, b\}} = baaba$. We say that the letters $a$ and $b$ alternate in $w$ if $w_{\{a, b\}}$ is either of the form $abababa\cdots$ or $bababab\cdots$ which can be of even or odd length. A word $w$ is called $k$-uniform if every letter occurs exactly $k$ times in $w$. A $1$-uniform word $w$ is a permutation on the set of letters of $w$. We say a word $u$ is a factor of a word $w$ if $w = xuy$ for some words $x$ and $y$, possibly empty words.

In this paper, we consider only simple (i.e., without loops or parallel edges) and connected graphs. A graph $G = (V, E)$, where $V$ is a set of vertices and $E$ is the set of edges, is called a word-representable graph if there is a word $w$ with the symbols of $V$ such that, for all $a, b \in V$, $a$ and $b$ are adjacent in $G$ if and only if $a$ and $b$ alternate in $w$. A word-representable graph $G$ is said to be $k$-word-representable if a $k$-uniform word represents it. It is known that every word-representable graph is $k$-word-representable, for some $k$ \cite{MR2467435}. The smallest number $k$ such that a graph $G$ is $k$-word-representable is called the representation number of $G$, denoted by $\mathcal{R}(G)$. 

A word-representable graph $G$ is said to be permutationally representable if it can be represented by a concatenation of some number of permutations on the vertices of $G$. If $G$ is represented by $p_1p_2 \cdots p_k$, where each $p_i$ is a permutation on the vertices of $G$, then we say $G$ is a permutationally $k$-representable graph. The permutation-representation number (in short \textit{prn}) of $G$, denoted by $\mathcal{R}^p(G)$, is the smallest number $k$ such that $G$ is permutationally $k$-representable. It is clear that $\mathcal{R}(G) \le \mathcal{R}^p(G)$. Note that for the complete graph $K_n$ on $n$ vertices, $\mathcal{R}^p(K_n) = 1 = \mathcal{R}(K_n)$. The class of word-representable graphs as well as the class of comparability graphs are hereditary with respect to the induced subgraphs (cf. \cite{words&graphs}). Moreover, if $H$ is an induced subgraph of a word-representable graph $G$, then $\mathcal{R}(H) \le \mathcal{R}(G)$; further, if $G$ is a comparability graph, then $\mathcal{R}^p(H) \le \mathcal{R}^p(G)$.

To distinguish between two-letter words and edges of graphs, we write $\overline{ab}$ to denote an undirected edge between vertices $a$ and $b$. A directed edge from $a$ to $b$ is denoted by $\overrightarrow{ab}$. An orientation of a graph is an assignment of direction to each edge so that the resulting graph is a directed graph. A vertex $a$ of a graph is said to be a source with respect to an orientation, if the in-degree of $a$ is zero. Similarly, if the out-degree of $a$ is zero, then $a$ is called a sink with respect to the orientation. Note that a source with respect to an orientation of a graph can be seen as a sink with respect to the orientation obtained by reversing the direction of each edge, and vice versa. An orientation of a graph $G = (V, E)$ is said to be a transitive orientation if $\overrightarrow{ab}$ and $\overrightarrow{bc}$, then $\overrightarrow{ac}$ with respect to the orientation, for all $a, b, c \in V$. If a graph admits a transitive orientation, then it is called a comparability graph. It is known that a graph is permutationally representable if and only if it is a comparability graph \cite{perkinsemigroup}. 

A set $P$ with a partial order, say $<$ on $P$, is called a partially ordered set or simply a poset. Two elements $a, b \in P$ are said to be comparable if $a < b$ or $b < a$; otherwise, they are called incomparable. An element $a \in P$ is called a maximal element if there is no element $b \in P$ such that $a < b$ in $P$. Similarly, an element $a$ is called minimal in $P$, if there is no element $b$ such that $b < a$ in $P$. For every poset $P$ there is a corresponding comparability graph whose vertices are the elements of $P$ and two vertices are adjacent if and only if they are comparable in $P$. Conversely, every comparability graph $G$ induces a poset, denoted by $P_G$, based on a given transitive orientation on $G$.   

We now recall the notion of the dimension of a poset from \cite{PosetDimension} and relate it with the \textit{prn} of the corresponding comparability graph. A partial order on a set $P$ is said to be a linear order if any two elements are comparable. A realizer of a poset $P$ is a collection of linear orders $\{L_1, L_2, \ldots, L_k\}$ on $P$ such that for every $a, b \in P$, $a < b$ in $P$ if and only if $a < b$ in $L_i$, for all $1 \le i \le k$. The  dimension of a poset $P$, denoted by $\dim(P)$, is the smallest positive integer $k$ such that $P$ has a realizer of size $k$. It was shown in \cite{MR2914710} that a comparability graph $G$ is permutationally $k$-representable if and only if the poset $P_G$ has dimension at most $k$. Consequently,  
if $G$ is a comparability graph, then $\mathcal{R}^p(G) = k$ if and only if $\dim(P_G) = k$ (cf. \cite{khyodeno2}).

We use the following notations in this paper. Let $G = (V, E)$ be a graph. The neighborhood of a vertex $a \in V$ is denoted by $N_G(a)$, and is defined by $N_G(a) = \{b \in V \mid \overline{ab}\in E\}$. For $A \subseteq V$, the neighborhood of $A$, $N_G(A) = \bigcup_{a \in A} N_G(a) \setminus A$. Further, the subgraph of $G$ induced by $A$ is denoted by $G[A]$.

We recall the concepts of split decomposition of a connected graph and its inverse operation, viz., split recomposition. A split of a connected graph $G = (V, E)$ is a bipartition $\{V_1, V_2\}$ of $V$ (i.e., $V = V_1 \cup V_2$ and $V_1 \cap V_2 = \emptyset$) satisfying the following: (i) $|V_1| \geq 2$ and $|V_2| \geq 2$, (ii) Each $x \in N_G(V_1)$ is adjacent to every $y \in N_G(V_2)$. If a graph has no split, then it is said to be a prime graph. 

A split decomposition of a graph $G = (V, E)$ with split $\{V_1, V_2\}$  is represented as a disjoint union of the induced subgraphs $G[V_1]$ and $G[V_2]$ along with an edge $e = \overline{v_1v_2}$, where $v_1$ and $v_2$ are two new vertices such that $v_1$ and $v_2$ are adjacent to each vertices of $N_G(V_2)$  and $N_G(V_1)$, respectively. By deleting the edge $e$, we obtain two components with vertex sets $V_1 \cup \{v_1\}$ and $V_2 \cup \{v_2\}$ called the split components. The two components are then decomposed recursively to obtain a split decomposition of $G$. 

Note that each split component of a graph $G$ is isomorphic to an induced subgraph of $G$ \cite{cicerone1999extension}. A minimal split decomposition of a graph is a split decomposition whose split components can be cliques, stars and prime graphs such that the number of split components is minimized. While there can be multiple split decompositions of a graph, a minimal split decomposition of a graph is unique \cite{Cunningham_2,Cunningham_1}. 

Let $G = (V \cup \{m\}, E)$ and $G' = (V' \cup \{m'\}, E')$ be two graphs with disjoint vertex sets. Then the split recomposition of $G$ and $G'$ with respect to the vertices $m$ and $m'$, called the marked vertices, is denoted by $G \srec G'$, and defined by $G \srec G' = (V'', E'')$, where the vertex set $V'' = V \cup V'$ and edge set $E''$ contains edges of $G[V]$, edges of $G'[V']$ and $\{\overline{ab} \mid a \in N_{G}(m), b \in N_{G'}(m')\}$.

\begin{remark}\label{com_reco}
	Let $G = (V \cup \{m\}, E)$ and $G' = (V' \cup \{m'\}, E')$ be two graphs. Then $G$ and $G'$ are induced subgraphs of $G \srec G'$. For instance, if $b \in N_{G}(m)$ and $b' \in N_{G'}(m')$, then $(G \srec G')[V \cup \{b'\}]$ is isomorphic to $G$, and $G'$ is isomorphic to $(G \srec G')[V' \cup \{b\}]$.
\end{remark}

\section{Word-Representable Graphs}

As stated in the following remark, while word-representability is preserved by split decomposition, in this section, we show that split recomposition also preserves the word-representability of graphs. Consequently, we observe that the class of parity graphs is word-representable. We also relate the representation number of the graph obtained by the recomposition with its components. 

\begin{remark}\label{Remark_i}
	If a graph $G$ is word-representable, then its split components are word-representable, as they are induced subgraphs of $G$. 
\end{remark}

\begin{theorem}\label{Theorem_1}
	If $G = (V \cup \{m\}, E)$ and $G' = (V' \cup \{m'\}, E')$ are two $k$-word-representable graphs, then so is the split recomposition $G \srec G'$.
\end{theorem}

\begin{proof}
Let $w$ and $w'$ be two $k$-uniform words representing the graphs $G$ and $G'$, respectively. In view of \cite[Proposition 5]{MR2467435}, note that a cyclic shift of a $k$-uniform word representing a graph also represents the graph. Accordingly, assume $w = w_1mw_2m \cdots w_km$ and $w' = m'w'_1m'w'_2 \cdots m'w'_k$.  Note that $w_V = w_1w_2 \cdots w_k$ and $w'_{V'} = w'_1w'_2 \cdots w'_k$. Moreover, $w_V$ and $w'_{V'}$ are $k$-uniform words. We will show that the $k$-uniform word $u = w_1w'_1w_2w'_2 \cdots w_kw'_k$ represents the graph $G \srec G'$.
	
Note that the induced subgraphs $G[V]$ and $G'[V']$ are induced subgraphs of $G \srec G'$. Further, the $k$-uniform words $w_V$ and $w'_{V'}$ represent $G[V]$ and $G'[V']$, respectively. Accordingly, since $u_V = w_V$ and $u_{V'} = w'_{V'}$, any two vertices of $G[V]$ or of $G'[V']$ are adjacent if and only if they alternate in the word $u$. 

For $a \in V$ and $b \in V'$, note that $a$ and $b$ are adjacent in $G \srec G'$ if and only if $a \in N_G(m)$ and $b \in N_{G'}(m')$. Accordingly, we consider the following cases to complete the proof.  

Let $a \in N_G(m)$ and $b \in N_{G'}(m')$. Note that $w_{\{a, m\}} =\; \stackrel{k}{am\; \cdots\; am}$ (i.e., the word obtained by concatenating $am$ for $k$ times) and $w'_{\{b, m'\}} =\; \stackrel{k}{m'b\; \cdots\; m'b}$. Accordingly, for all $1 \le i \le k$, $w_i$ and $w'_i$ contain exactly one copy of $a$ and $b$, respectively. Hence, $u_{\{a, b\}} =\; \stackrel{k}{ab\; \cdots\; ab}$ so that $a$ and $b$ alternate in $u$. 

Let $a \notin N_{G}(m)$. Since the $k$-uniform word $w$ represents $G$, the vertices $a$ and $m$ do not alternate in $w$. Hence, $aa$ is a substring of $w_t$, for some $t$. For any $b \in V'$, note that $b$ does not appear in $w_t$. According to the construction of $u$, the subword $u_{\{a, b\}}$ has $aa$ as a factor. Hence, $a$ and $b$ do not alternate in $u$. Similarly, when  $b \notin N_{G'}(m')$, we can observe that $b$ does not alternate with any $a \in V$. 

Hence, the $k$-uniform word $u$ represents $G \srec G'$. \qed	
\end{proof}

\begin{corollary}
	If $G = (V \cup \{m\}, E)$ and $G' = (V' \cup \{m'\}, E')$ are two word-representable graphs such that $\mathcal{R}(G) = k$ and $\mathcal{R}(G') = k'$, then $\mathcal{R}(G \srec G') = \max\{k, k'\}$.
\end{corollary}

\begin{proof}
	Let $t = \max\{k, k'\}$. Using \cite[Observation 4]{MR2467435}, both $G$ and $G'$ are $t$-word-representable. Consequently, by Theorem \ref{Theorem_1},  $G \srec G'$ is $t$-word-representable so that $\mathcal{R}(G \srec G') \leq t$. But, by Remark \ref{com_reco}, since $G$ and $G'$ are isomorphic to certain induced subgraphs of $G \srec G'$, we have $\mathcal{R}(G \srec G') \geq t$. Hence, $\mathcal{R}(G \srec G') = \max\{k, k'\}$. \qed	
\end{proof}

We now summarize the characterization of word-representability of a graph with respect to its split decomposition in the following theorem. 

\begin{theorem}\label{Theorem_2}
	Let $H_i$ $(1 \le i \le s)$ be the components of a split decomposition of a graph $G$. Then $G$ is word-representable if and only if $H_i$ is word-representable for all $1 \le i \le s$. Further,  $\mathcal{R}(G) = \displaystyle\max_{1 \le i \le s} \mathcal{R}(H_i)$.
\end{theorem}

In view of the fact that the class of perfect graphs is closed under split recomposition \cite{comp_perfectgraphs}, we have the following corollary of Theorem \ref{Theorem_2}. Although a characterization for word-representable perfect graphs is an open problem, this corollary may be useful in this regard.

\begin{corollary} \label{perfectgraphs}
	 A perfect graph is not word-representable if and only if at least one of its prime components is a non-word-representable perfect graph. 
\end{corollary}
Thus, the class of prime perfect graphs that are not word-representable completely determines the class of  perfect graphs that are not word-representable. 

A graph is called a parity graph if the lengths of any two induced paths between a pair of vertices are of the same parity. The class of parity graphs is a subclass of perfect graphs \cite{sachs}.  Recall from \cite{cicerone1999extension} that a graph $G$ is a parity graph if and only if every component in its minimal split decomposition is a clique or a bipartite graph.  Accordingly, since cliques and bipartite graphs are word-representable, we have the following corollary of Theorem \ref{Theorem_2}. Thus, we establish a class of word-representable graphs within perfect graphs which includes bipartite graphs as well as distance-hereditary graphs. 

\begin{corollary}\label{parity_graphs}
The class of parity graphs is word-representable.
\end{corollary}

\section{Comparability Graphs}

In this section, we study the recomposition of comparability graphs. First we observe that the class of comparability graphs is not closed under split recomposition. In Subsection 4.1, we give a sufficient condition so that a recomposition is a comparability graph. Further, in Subsection 4.2, we investigate the \textit{prn} of a recomposition of comparability graphs. Finally, in Subsection 4.3, we provide a characteristic property by which a recomposition of comparability graphs is a comparability graph.

Let $G$ be a comparability graph. Note that every split component of $G$ is a comparability graph, as it is (isomorphic) an induced subgraph of $G$. However, as shown in the following example, split recomposition of two comparability graphs is not always a comparability graph.

\begin{example}\label{ex_1}
	The graphs $G$ and $G'$ shown in Fig. \ref{fig_1} are comparability graphs. However, the split recomposition $G \srec G'$ is not a comparability graph (cf. \cite[Chapter 5]{MR2063679}). 
	\begin{figure}[ht]
		\centering
		\begin{minipage}{.4\textwidth}
			\centering
			\[\begin{tikzpicture}[scale=0.6]
			\vertex (a_8) at (-3.4,0) [fill=black, label=below:$1$] {};
			\vertex (a_7) at (-2.4,0) [fill=black, label=below:$2$] {};
			\vertex (a_6) at (-1.4,0) [label=below:$m$] {};
			\vertex (a_1) at (1,0) [label=below:$m'$] {};
			\vertex (a_4) at (2,1.5) [fill=black, label=right:$3$] {};
			\vertex (a_2) at (2,0.5) [fill=black, label=right:$4$] {};
			\vertex (a_3) at (2,-0.5) [fill=black, label=right:$5$] {};
			\vertex (a_5) at (2,-1.5) [fill=black, label=right:$6$] {};
			\node (a_9) at (-2.4,-1.5) [label=below:$G$] {};
			\node (a_{10}) at (2,-1.5) [label=below:$G'$] {};
			
			\path
			(a_1) edge (a_2)
			(a_1) edge (a_3)
			(a_2) edge (a_3)
			
			(a_2) edge (a_4)
			(a_3) edge (a_5)
			(a_6) edge (a_7)
			(a_7) edge (a_8) ;
			\end{tikzpicture}\] 
		\end{minipage}%
		\begin{minipage}{.3\textwidth}
			\centering
			
			\[\begin{tikzpicture}[scale=0.6]
			\vertex (a_1) at (-0.4,0) [fill=black,label=below:$2$] {};
			\vertex (a_2) at (1,0.5) [fill=black,label=right:$4$] {};
			\vertex (a_3) at (1,-0.5) [fill=black,label=right:$5$] {};
			\vertex (a_4) at (1,1.5) [fill=black,label=right:$3$] {};
			\vertex (a_5) at (1,-1.5) [fill=black,label=right:$6$] {};
			\vertex (a_6) at (-1.4,0) [fill=black,label=below:$1$] {};
			\node (a_{10}) at (-0.4,-1.5) [label=below:$G \srec G'$] {};
			
			\path
			(a_1) edge (a_2)
			(a_1) edge (a_3)
			(a_2) edge (a_3)
			(a_1) edge (a_6)
			(a_2) edge (a_4)
			(a_3) edge (a_5);
			\end{tikzpicture}\]
		\end{minipage}%
		
		\caption{Recomposition of two comparability graphs} 
		\label{fig_1}
	\end{figure}
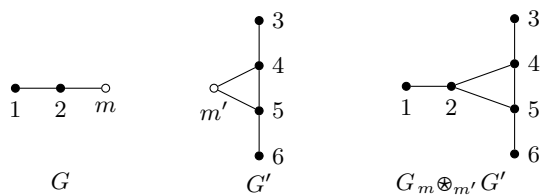
	
\end{example}

For the rest of the paper, let $G = (V \cup \{m\}, E)$ and $G' = (V' \cup \{m'\}, E')$ be two graphs with disjoint vertex sets. Consider the graph $$H = (V \cup V' \cup \{m, m'\}, E \cup E' \cup \{\overline{mm'}\}),$$ which is obtained by adding the edge $\overline{mm'}$ between $G$ and $G'$. This graph operation is useful in the sequel for investigating a sufficient condition and the \textit{prn} of $G \srec G'$.

\subsection{Sufficient Condition}

In this subsection, we give a necessary and sufficient condition for $H$ to be a comparability graph. Further, we also discuss the conditions under which the recomposition $G \srec G'$ is a comparability graph. As a consequence, we also establish that a recomposition of bipartite graphs is a bipartite graph.

\begin{theorem} \label{lemma_1}
	 The graph $H$ is a comparability graph if and only if there exist transitive orientations $T$ and $T'$ of $G$ and $G'$, respectively, such that $m$ is a source of $T$ and $m'$ is a sink of $T'$, or vice versa.
\end{theorem}

\begin{proof}
	Suppose $T$ and $T'$ are transitive orientations of $G$ and $G'$, respectively, such that  $m$ is a source of $T$ and $m'$ is a sink of $T'$. We extend $T$ and $T'$ to an orientation $T_H$ of $H$ by orienting the edge $\overline{mm'}$ as $\overrightarrow{mm'}$. In case, if $m$ is a sink of $T$ and $m'$ is a source of $T'$, then assign $\overrightarrow{m'm}$ as the orientation to the edge $\overline{mm'}$ in $T_H$. In both the cases, clearly, $T_H$ is a transitive orientation of $H$ so that $H$ is a comparability graph.  
	 
	Conversely, suppose $H$ is a comparability graph. Let $T_H$ be a transitive orientation of $H$ in which the edge $\overline{mm'}$ is oriented as $\overrightarrow{mm'}$. We claim that $m$ is a source and $m'$ is a sink with respect to $T_H$. On the contrary, we have the following cases:
	 \begin{itemize}
	 	\item $m$ is not a source: Then, for some $a \in V$,  we have $\overrightarrow{am}$ in $T_H$. By transitivity of $T_H$, $\overrightarrow{am'}$ exists in $T_H$; which is a contradiction as $\overline{am'}$ is not an edge in $H$. 
	 	\item $m'$ is not a sink: Then, for some $b \in V'$,  we have $\overrightarrow{m'b}$ in $T_H$. By transitivity of $T_H$, $\overrightarrow{mb}$ exists in $T_H$; which is a contradiction as $\overline{mb}$ is not an edge in $H$.  
	\end{itemize}
	This proves our claim. Note that $H[V \cup \{m\}] = G$ and $H[V' \cup \{m'\}] = G'$. Further, let $T$ and $T'$ be the restrictions of $T_H$ to $G$ and $G'$, respectively. Clearly, $T$ and $T'$ are transitive orientations of $G$ and $G'$, respectively. Moreover, $m$ is a source of $T$ and $m'$ is a sink of $T'$. 
 
 Similarly, if the edge $\overline{mm'}$ of $H$ is oriented as $\overrightarrow{m'm}$ in $T_H$, then we can show that there exist transitive orientations $T$ and $T'$ of $G$ and $G'$, respectively, such that $m$ is a sink of $T$ and $m'$ is a source of $T'$. 
\qed	
\end{proof}

We now present a condition under which a recomposition of comparability graphs is a comparability graph.

\begin{theorem}\label{rec_comparability}
	If $G$ and $G'$ are comparability graphs such that $m$ is a source of $G$ and $m'$ is a sink of $G'$ with respect to some transitive orientations of $G$ and $G'$,  then $G \srec G'$ is a comparability graph.
\end{theorem}

\begin{proof}
	By Theorem \ref{lemma_1}, $H$ is a comparability graph in which $m$ is a source and $m'$ is a sink. In view of the discussion in \cite[Section 3]{cornelsen2009treelike}, we give a transitive orientation to $G \srec G'$ from a transitive orientation of $H$ by deleting $m$ and $m'$ and by orienting $\overrightarrow{ba}$ for all adjacent vertices $a (\neq m')$ of $m$ and $b (\neq m)$ of $m'$.  	\qed
\end{proof}

\begin{remark}
	In Theorem \ref{rec_comparability}, if $m$ is a sink and $m'$ is a source with respect to some transitive orientations of $G$ and $G'$, respectively, then also a transitive orientation can be assigned to $G \srec G'$. Hence, in the rest of the paper, we consider one of the combinations to state the corresponding results.
\end{remark}	

Since every bipartite graph is a comparability graph, we have the following corollary.

\begin{corollary}\label{bipartite_recom}
	 If $G$ and $G'$ are two bipartite graphs, then for any choice of vertices $m$ and $m'$, the split recomposition $G \srec G'$ is a comparability graph.
\end{corollary}

\begin{proof}
	Let $\{A, B\}$ be a bipartition of $G$ and $\{A', B'\}$ be a bipartition of $G'$. Without loss of generality, suppose $m \in A$ and $m' \in A'$. Assign orientation to the edges in $E$ as $\overrightarrow{ab}$ such that $a \in A, b \in B$ and the edges in $E'$ as $\overrightarrow{b'a'}$ such that $b' \in B', a' \in A'$. In this transitive orientation, $m$ is a source and $m'$ is a sink. Hence, the result follows from Theorem \ref{rec_comparability}. \qed
\end{proof}

Moreover, we have the following result for bipartite graphs.

\begin{proposition}\label{lemma_2}
	If $G$ and $G'$ are two bipartite graphs, then for any choice of vertices $m$ and $m'$, the split recomposition $G \srec G'$ is a bipartite graph.
\end{proposition}

\begin{proof}
	Let $\{A, B\}$ be a bipartition of $G$ and $\{A', B'\}$ be a bipartition of $G'$. Without loss of generality, suppose $m \in A$ and $m' \in A'$.
	We prove that $\{(A \setminus \{m\}) \cup B', (A' \setminus \{m'\}) \cup B\}$ is a bipartition of the vertices of $G \srec G'$. Consider an edge $\overline{ab}$ of $G \srec G'$. Then $\overline{ab}$ is in one of the induced subgraphs $G[V]$ and $G'[V']$, or it is an edge between $N_{G}(m)$ and $N_{G'}(m')$. In any case, it can be observed that if $a \in (A \setminus \{m\}) \cup B'$ then $b \in (A' \setminus \{m'\}) \cup B$, or vice versa.\qed
\end{proof}
 
\subsection{Permutation-Representation Number} 
 
In this subsection, we study the \textit{prn} of the graph $H$ and obtain a partial result on the \textit{prn} of $G \srec G'$. A complete result on the \textit{prn} of $G \srec G'$ is given in Subsection 4.3, subsequent to its characterization.  

\begin{lemma}\label{lemma_3}
	Let $G$ be a comparability graph with a transitive orientation $T$. If $m$ is a source with respect to $T$, then $m$ is a minimal element of the induced poset $P_G$. Similarly, if $m$ is a sink, then $m$ is a maximal element of $P_G$.   
\end{lemma}

\begin{proof}
	Let us assume that $m$ is not a minimal element of $P_G$. Then there exists $a$ in $P_G$ such that $a < m$ in $P_G$ which implies there is an edge $\overrightarrow{am}$ in $G$ with respect to $T$. This is a contradiction to $m$ is a source. Similarly, the result can be proved when $m$ is a sink. 
	\qed
\end{proof}

The following remark produces a word to represent a comparability graph based on a realizer of its induced poset.

\begin{remark}\label{per_lin_ext}
	Let $C$ be a comparability graph and $P_{C}$ be the poset induced by a  transitive orientation of $C$. Let $\{L_1, L_2, \ldots, L_k\}$ be a realizer of $P_C$. For $1 \le i \le k$, construct permutations $p_i$ from $L_i$ on the vertices of $C$ such that $a$ occurs before $b$ in the permutation $p_i$ if $a < b$ in $L_i$. Then, in view of \cite[Lemma 4]{MR2914710}, the word $ p_1p_2 \cdots p_k$ represents $C$.
\end{remark}

\begin{proposition}\label{prop_perm_ext}
	Let $G$ be a comparability graph such that $m$ is a source (or sink) with respect to some transitive orientation of $G$. Then there exist $k$ and permutations $p_i$ on $V$, for $1 \le i \le k$, which can be extended to permutations $q_i$ on $V \cup \{m\}$, for $1 \le i \le k+1$, such that $p_1\cdots p_k$ represents $G[V]$ and $q_1 \cdots q_kq_{k+1}$ represents $G$.
\end{proposition}

\begin{proof}
	Consider the poset $P_G$ induced by the transitive orientation of $G$ for which $m$ is a source. Note that $m$ is a minimal element of $P_{G}$. Now consider the subposet of $P_{G}$ with the elements of $V$ and note that the corresponding comparability graph is $G[V]$. Accordingly, we denote the subposet by $P_{G[V]}$.
	
	Consider a realizer $L = \{L_1, L_2, \ldots, L_k\}$ of $P_{G[V]}$. Using the construction in \cite[Theorem 4.2]{Hiraguchiresult} (see also \cite[Theorem 9.3 in Chapter 1]{Trotterbook} and \cite[Theorem 2.11]{bogart1973maximal}), we extend $L$ to a realizer $M = \{M_1, M_2, \ldots, M_k, M_{k+1}\}$ of $P_{G}$ as per the following:	Let $U(m)= \{a \in V \mid a > m  \text{ in }  P_{G}\}$ and set
	\begin{align*}
		M_i &= m < L_i, \quad (1 \le i \le k); \\
		M_{k+1} &= L_k(V \setminus U(m)) < m < L_k(U(m)),
	\end{align*}
	where, for $X \subseteq V$, $L_k(X)$ is the linear order on $X$ obtained from $L_k$ by restricting it to $X$.     
	
	For $1 \le i \le k$, construct permutations $p_i$ from $L_i$ such that $a$ occurs before $b$ in $p_i$ if $a < b$ in $L_i$. Similarly, for $1 \le i \le k+1$, construct permutations $q_i$  from $M_i$. Note that, as $m$ is a minimal element of $P_{G}$, we have $N_{G}(m) = U(m)$. Accordingly, note that 	
	\begin{align*}
		q_{i} &= mp_{i}, \quad(1 \le i \le k); \\
		q_{k + 1} &= p_{{k}_{V \setminus N_{G}(m)}}m p_{{k}_{N_{G}(m)}}.	
	\end{align*}
	Hence, by Remark \ref{per_lin_ext}, the words $p_1p_2 \cdots p_{k}$ and $q_1q_2 \cdots q_{k}q_{k + 1}$ represent the graphs $G[V]$ and  $G$, respectively. 
	
	Similarly, if $m$ is a sink, there exist permutations $p_i$ on $V$, for $1 \le i \le k$, and permutations $q_i$ on $V \cup \{m\}$, for $1 \le i \le k+1$, such that the words $p_1p_2 \cdots p_{k}$ and  $q_1q_2 \cdots q_{k}q_{k + 1}$ represent $G[V]$ and $G$, respectively, where the permutations $q_i$ are as per the following: 	
	\begin{align*}
		q_{i} &= p_{i}m, \quad(1 \le i \le k); \\
		q_{k + 1} &= p_{{k}_{N_{G}(m)}}mp_{{k}_{V \setminus  N_{G}(m)}}.
	\end{align*}
	\qed
\end{proof}

Note that $\mathcal{R}^p(G) = k$ implies $\dim(P_G) = k$ so that $\dim(P_{G[V]}) \le k$. Accordingly, Proposition \ref{prop_perm_ext} is rephrased in the following lemma for the graphs $G$ and $G'$. The lemma is useful for obtaining the \textit{prn} of their recomposition and also that of $H$.

\begin{lemma}\label{const_perm_GG'}
		Suppose $G$ and $G'$ are comparability graphs such that $m$ is a source of $G$ and $m'$ is a sink of $G'$ with respect to some transitive orientations of $G$ and $G'$. Let $\mathcal{R}^p(G) =k$ and $\mathcal{R}^p(G') =k'$. Then there exist permutations $p_i$  $(1 \le i \le k)$ on $V$, $p'_i$  $(1 \le i \le k')$ on $V'$ which can be extended to the permutations $q_i$ $(1 \le i \le k+1)$ on the vertices of $G$ and permutations $q'_i$ $(1 \le i \le k'+1)$ on the vertices of $G'$ such that the words $p_1p_2 \cdots p_{k}$ and $p'_1p'_2 \cdots p'_{k'}$ represent $G[V]$ and $G'[V']$, respectively and the words $w = q_1q_2 \cdots q_{k}q_{k + 1}$ and $w' = q'_1q'_2 \cdots q'_{k'}q'_{k' + 1}$ represent the graphs $G$ and $G'$, respectively, where		
		\[\begin{array}{rlrl}
			q_{i} &= mp_{i}, \quad(1 \le i \le k); \quad \quad & q'_{i} &= p'_{i}m', \quad(1 \le i \le k'); \\
			q_{k + 1} &= p_{{k}_{V \setminus N_{G}(m)}}m p_{{k}_{N_{G}(m)}}; \quad \quad & q'_{k' + 1} &= p'_{{k'}_{N_{G'}(m')}}m'p'_{{k'}_{V' \setminus  N_{G'}(m')}}. \\
		\end{array}\]
\end{lemma}

	\begin{theorem}\label{lemma_4}
		Let $G$ and $G'$ be comparability graphs such that  $\mathcal{R}^p(G) = k$ and $\mathcal{R}^p(G') = k'$.  If the graph $H$ is a comparability graph, then $\mathcal{R}^p(H) = \max\{k, k'\}$ or $1 + \max\{k, k'\}$.
	\end{theorem}

\begin{proof}
	Suppose $H$ is a comparability graph. By Theorem \ref{lemma_1}, $m$ is a source and  $m'$ is a sink with respect to some transitive orientations of $G$ and $G'$, respectively.
	
	As $\mathcal{R}^p(G) = k$ and $\mathcal{R}^p(G') = k'$, using Lemma \ref{const_perm_GG'}, consider the permutations $p_i$ $(1 \le i \le k)$ on $V$, $p'_i$ $(1 \le i \le k')$ on $V'$, $q_i$ $(1 \le i \le k+1)$ on the vertices of $G$, and $q'_i$ $(1 \le i \le k'+1)$ on the vertices of $G'$ such that the words $w = q_1q_2 \cdots q_{k}q_{k + 1}$ and $w' = q'_1q'_2 \cdots q'_{k'}q'_{k' + 1}$ represent the graphs $G$ and $G'$, respectively.
	
	Suppose $k' \le k$. We construct $k+1$ number of permutations $u_i$ ($1 \le i \le k+1$) on the vertices of $H$ as per the following:
	\begin{align*}
		u_{1} & = p'_1mp_1m', \\
		u_{i} & = p'_imm'p_i, \quad(2 \le i \le k') \\
		u_{j} & = p'_{k'}mm'p_{j}, \quad(k' + 1 \le j \le k) \\
		u_{k + 1} & = q_{k + 1}q'_{k' + 1} \\
		& =p_{{k}_{V \setminus N_{G}(m)}}m p_{{k}_{N_{G}(m)}}p'_{{k'}_{N_{G'}(m')}}m'p'_{{k'}_{V' \setminus  N_{G'}(m')}}
	\end{align*}
	
	We now show that the word $u = u_1u_2 \cdots u_{k}u_{k + 1}$ represents the graph $H$.
	Note that $H[V \cup \{m\}] = G$ and $H[V' \cup \{m'\}] = G'$. Further, since	
	\begin{align*}
		u_{V \cup \{m\}} & = q_1q_2\; \cdots \; q_kq_{k+1}, \text{and} \\
		u_{V' \cup \{m'\}} & = q'_1q'_2\; \cdots \;q'_{k'-1} \overbrace{q'_{k'}\; \cdots \; q'_{k'}}^{k-k'+1}q'_{k'+1},
	\end{align*}
     we have $u_{V \cup \{m\}}$ and $u_{V' \cup \{m'\}}$ represent the graphs $G$ and $G'$, respectively. Hence, any two vertices of $G$ or of $G'$ are adjacent if and only if they alternate in the word $u$. 
	Also, corresponding to the edge $\overline{mm'}$ in $H$,  it is easy to see that $m$ and $m'$ alternate in $u$, as $mm' \preceq u_i$ for all $1 \le i \le k + 1$. 
	
	Note that there are no edges between the graphs $G$ and $G'$ except $\overline{mm'}$; accordingly, we show that the corresponding vertices do not alternate in $u$.
	For $a \in V$ and $a' \in V' \cup \{m'\}$, as $a'a \preceq u_2$ but $aa' \preceq u_{k + 1}$, it is evident that $a$ and $a'$ do not alternate in $u$. Also, for $a \in V \cup \{m\}$ and $a' \in V'$, we have $a$ and $a'$ do not alternate in $u$ as $a'a \preceq u_1$ but $aa' \preceq u_{k + 1}$. Hence, the word $u$ represents the graph $H$.
	
    Similarly, when $k \le k'$, we can show that the word $u = u_1u_2 \cdots u_{k'}u_{k' + 1}$ represents the graph $H$, where
	
	\begin{align*}
		u_{1} & = p'_1mp_1m'; \\
		u_{i} & = p'_imm'p_i, \quad(2 \le i \le k); \\
		u_{j} & = p'_{j}mm'p_{k}, \quad(k + 1 \le j \le k'); \\
		u_{k' + 1} & = q_{k + 1}q'_{k' + 1} \\
		& =p_{{k}_{V \setminus N_{G}(m)}}m p_{{k}_{N_{G}(m)}}p'_{{k'}_{N_{G'}(m')}}m'p'_{{k'}_{V' \setminus  N_{G'}(m')}}.
	\end{align*}

	 In any case, since $u$ is a concatenation of $1 + \max\{k, k'\}$ number of permutations, we have $\mathcal{R}^p(H) \le 1 + \max\{k, k'\}$. Also, since $G$ and $G'$ are induced subgraphs of $H$, we have  $\mathcal{R}^p(H) \ge \max\{k, k'\}$. Hence, $\mathcal{R}^p(H) = \max\{k, k'\}$ or $1 + \max\{k, k'\}$. \qed 
\end{proof}

\begin{theorem}\label{theorem_7}
	Let $G$ and $G'$ be comparability graphs such that $m$ is a source of $G$ and $m'$ is a sink of $G'$ with respect to some transitive orientations of $G$ and $G'$. If $\mathcal{R}^p(G) = k$ and $\mathcal{R}^p(G') = k'$, then $\mathcal{R}^p(G \srec G') = \max\{k, k'\}$ or $1+ \max\{k, k'\}$.
\end{theorem}

\begin{proof}
	 Let us assume that  $m$ is a source of $G$ and $m'$ is a sink of $G'$ with respect to some transitive orientations of $G$ and $G'$. Then, by Theorem \ref{rec_comparability}, $G \srec G'$ is a comparability graph and hence permutationally representable. As $\mathcal{R}^p(G) = k$ and $\mathcal{R}^p(G') = k'$, using Lemma \ref{const_perm_GG'}, consider the permutations $p_i$ $(1 \le i \le k)$ on $V$, $p'_i$ $(1 \le i \le k')$ on $V'$, $q_i$ $(1 \le i \le k+1)$ on the vertices of $G$, and $q'_i$ $(1 \le i \le k'+1)$ on the vertices of $G'$ such that the words $w = q_1q_2 \cdots q_{k}q_{k + 1}$ and $w' = q'_1q'_2 \cdots q'_{k'}q'_{k' + 1}$ represent the graphs $G$ and $G'$, respectively.
	
	Suppose $k \le k'$. We construct $k' +1 $ number of permutations $v_i$ ($1 \le i \le k' + 1$) on the vertices of $G \srec G'$ as per the following:  
	
	\begin{align*}
		v_{i} & = p'_ip_i, \quad(1 \le i \le k) \\
		v_{j} & = p'_jp_{k}, \quad(k + 1 \le j \le k') \\
		v_{k' + 1} & = p_{{k}_{V \setminus N_{G}(m)}}p'_{{k'}_{N_{G'}(m')}}p_{{k}_{N_{G}(m)}}p'_{{k'}_{V' \setminus N_{G'}(m')}}
	\end{align*}
	
   We now show that the word $v = v_1v_2 \cdots v_{k'}v_{k'+1}$ represents the graph $G \srec G'$.	Note that $G[V]$ and $G'[V']$ are induced subgraphs of $G \srec G'$ represented by
   \begin{align*}
     	w_{V} & =  p_1p_2 \; \cdots \; p_{k} p_{{k}_{V \setminus N_{G}(m)}}p_{{k}_{N_{G}(m)}}, \text{and} \\
     	w'_{V'} & = p'_1p'_2 \; \cdots \; p'_{k'}p'_{{k'}_{N_{G'}(m')}}p'_{{k'}_{V' \setminus  N_{G'}(m')}},
   \end{align*}
      respectively. Further, note that
      \begin{align*}
      	v_{V} & = p_1p_2 \; \cdots \; p_{k-1} \overbrace{p_k \; \cdots \; p_k}^{k' - k +1}p_{{k}_{V \setminus N_{G}(m)}}p_{{k}_{N_{G}(m)}}, \text{and} \\
      	v_{V'} & = p'_1p'_2 \; \cdots \; p'_{k'}p'_{{k'}_{N_{G'}(m')}}p'_{{k'}_{V' \setminus  N_{G'}(m')}}.
      \end{align*}
      
      Thus, any two vertices of $G[V]$ or of $G'[V']$ are adjacent if and only if they alternate in the word $v$.  
	
	For $a \in V$ and $a' \in V'$, note that $a$ and $a'$ are adjacent in $G \srec G'$ if and only if $a \in N_G(m)$ and $a' \in N_{G'}(m')$. Accordingly, we consider the following cases to show that $a$ and $a'$ are adjacent in $G \srec G'$ if and only if they alternate in $v$. If $a \in N_G(m)$ and $a' \in N_{G'}(m')$, then from the construction of $v$, we have $a'a \preceq v_i$ for all $1 \le i \le k' + 1$, so that $a$ and $a'$ alternate in $v$. If $a \notin N_{G}(m)$, then for any $a' \in V'$, we have $a'a \preceq v_1$ but $aa' \preceq v_{k' + 1}$.  Similarly, if $a' \notin N_{G'}(m')$, then for any vertex $a \in V$, we have $a'a \preceq v_1$ but $aa' \preceq v_{k' + 1}$. Hence, $a$ and $a'$ do not alternate in $v$, if $a \notin N_G(m)$ or $a' \notin N_{G'}(m')$. Hence the word $v$ represents $G \srec G'$ permutationally.
	
	Similarly, when $k' \le k$, we can show that the word $v = v_1v_2 \cdots v_kv_{k+1}$ represents the graph $G \srec G'$ permutationally, where 
	\begin{align*}
		v_{i} & = p'_ip_i, \quad(1 \le i \le k'); \\
		v_{j} & = p'_{k'}p_{j}, \quad(k' + 1 \le j \le k); \\
		v_{k + 1} & = p_{{k}_{V \setminus N_{G}(m)}}p'_{{k'}_{N_{G'}(m')}}p_{{k}_{N_{G}(m)}}p'_{{k'}_{V' \setminus N_{G'}(m')}}.
	\end{align*}
	
	In any case, since $v$ is a concatenation of $1 + \max\{k, k'\}$ number of permutations, we have $\mathcal{R}^p(G \srec G') \le 1 + \max\{k, k'\}$. Also, since $G$ and $G'$ are isomorphic to certain induced subgraphs of $G \srec G'$ (cf. Remark \ref{com_reco}), we have  $\mathcal{R}^p(G \srec G') \ge \max\{k, k'\}$. Hence,
	$\mathcal{R}^p(G \srec G') = \max\{k, k'\}$ or $1+ \max\{k, k'\}$. \qed  
\end{proof}

 In view of Corollary \ref{bipartite_recom}, we have the following corollary of Theorem \ref{theorem_7}.

\begin{corollary}\label{coro_3}
	If $G$ and $G'$ are bipartite graphs with $\mathcal{R}^p(G) = k$ and $\mathcal{R}^p(G') = k'$,  then $\mathcal{R}^p(G \srec G') = \max\{k, k'\}$ or $1 + \max\{k, k'\}$.
\end{corollary}	

\subsection{Characterization}

In summary, we obtained a transitive orientation for a recomposition of comparability graphs when the marked vertices are source and sink with respect to their transitive orientations. However, the marked vertices need not be source or sink to give a transitive orientation to the recomposition, as shown in the following example.
 
\begin{example}\label{EX_2}
	Consider the comparability graphs $G_1$ and $G_2$ given in Fig. \ref{fig_2} in which $a$ and $b$ are the marked vertices. Note that $b$ is not a source or sink under any transitive orientation given to $G_2$. However, $G_1 \psrec{a}{b} G_2$ is a comparability graph. 
	
	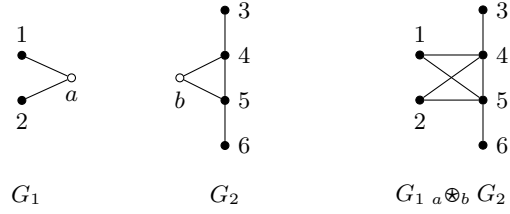
\begin{figure}[!h]
		\centering
		\begin{minipage}{.4\textwidth}
			\centering
			\[\begin{tikzpicture}[scale=0.6]
				\vertex (a_8) at (-2.5,0.5) [fill=black, label=above:$1$] {};
				\vertex (a_7) at (-2.5,-0.5) [fill=black, label=below:$2$] {};
				\vertex (a_6) at (-1.4,0) [label=below:$a$] {};
				\vertex (a_1) at (1,0) [label=below:$b$] {};
				\vertex (a_4) at (2,1.5) [fill=black, label=right:$3$] {};
				\vertex (a_2) at (2,0.5) [fill=black, label=right:$4$] {};
				\vertex (a_3) at (2,-0.5) [fill=black, label=right:$5$] {};
				\vertex (a_5) at (2,-1.5) [fill=black, label=right:$6$] {};
				\node (a_9) at (-2.4,-2) [label=below:$G_1$] {};
				\node (a_{10}) at (2,-2) [label=below:$G_2$] {};
				
				\path
				(a_1) edge (a_2)
				(a_1) edge (a_3)
				(a_2) edge (a_3)
				
				(a_2) edge (a_4)
				(a_3) edge (a_5)
				(a_6) edge (a_7)
				(a_6) edge (a_8) ;
			\end{tikzpicture}\] 
		\end{minipage}%
		\begin{minipage}{.3\textwidth}
			\centering
			\[\begin{tikzpicture}[scale=0.6]
				\vertex (a_1) at (-0.4,0.5) [fill=black,label=above:$1$] {};
				\vertex (a_2) at (1,0.5) [fill=black,label=right:$4$] {};
				\vertex (a_3) at (1,-0.5) [fill=black,label=right:$5$] {};
				\vertex (a_4) at (1,1.5) [fill=black,label=right:$3$] {};
				\vertex (a_5) at (1,-1.5) [fill=black,label=right:$6$] {};
				\vertex (a_6) at (-0.4,-0.5) [fill=black,label=below:$2$] {};
				\node (a_{10}) at (0.3,-2) [label=below:$G_1 \psrec{a}{b} G_2$] {};
				
				\path
				(a_1) edge (a_2)
				(a_1) edge (a_3)
				(a_2) edge (a_3)
				(a_2) edge (a_6)
				(a_3) edge (a_6)
				(a_2) edge (a_4)
				(a_3) edge (a_5);
			  \end{tikzpicture}\] 
		\end{minipage}%
		\caption{Recomposition of comparability graphs} 
		\label{fig_2}
	\end{figure}
\end{example}

However, in Example \ref{EX_2},  the vertex $a$ is an all-adjacent vertex, i.e., $a$ is adjacent to all other vertices of $G_1$. In connection to this observation we prove the following theorem.

\begin{theorem}\label{th_all_adjacent}
	Let $G$ and $G'$ be two comparability graphs such that  at least one of $m$ and $m'$ is an all-adjacent vertex. Then $G \srec G'$ is a comparability graph. Moreover, if $\mathcal{R}^p(G) = k$ and $\mathcal{R}^p(G') = k'$, then $\mathcal{R}^p(G \srec G') = \max\{k, k'\}$. 
\end{theorem}

\begin{proof}
	We prove the result by constructing requisite number of permutations on the vertices of $G \srec G'$ whose concatenation represents $G \srec G'$.
	
	Let the words $p_1p_2 \cdots p_k$ and $p'_1p'_2 \cdots p'_{k'}$ represent $G$ and $G'$, respectively, where each $p_i$ $(1 \leq i \leq k)$ is a permutation on the vertices of $G$ and each $p'_i$ $(1 \leq i \leq k')$ is a permutation on the vertices of $G'$. 
	
	Suppose $\max\{k, k'\} = t$. If $k' < k$, then set $p'_j = p'_{k'}$ for all $k'+ 1 \le j \le t$ and note that $p'_1p'_2 \cdots p'_t$ represents $G'$. Similarly, if $k < k'$, then set $p_j = p_{k}$ for all $k + 1 \le j \le t$ and note that $p_1p_2 \cdots p_t$ represents $G$. In any case, the words $w = p_1p_2 \cdots p_t$ and $w' = p'_1p'_2 \cdots p'_t$ represent the graphs $G$ and $G'$, respectively. For $1 \le i \le t$, let $p_i = r_ims_i$  and $p'_i = r'_im's'_i$ so that $p_{i_V} = r_is_i$ and $p'_{i_{V'}} = r'_is'_i$. 
	
	Suppose $m$ is an all-adjacent vertex of $G$. Let $u_i = r'_ir_is_is'_i$, for all $1 \le i \le t$. Note that each $u_i$ is a permutation on the vertices of $G \srec G'$.  We show that the word $u = u_1u_2 \cdots u_t$  represents the graph $G \srec G'$ permutationally.
	
	Note that the induced subgraphs $G[V]$ and $G'[V']$ are induced subgraphs of $G \srec G'$. Further, since $u_{V} = p_{1_V}p_{2_V} \cdots p_{t_V} = w_{V}$ and $u_{V'} = p'_{1_{V'}}p'_{2_{V'}} \cdots p'_{{t}_{V'}} = w'_{V'}$, the subwords $u_{V}$ and $u_{V'}$ of $u$ represent the graphs $G[V]$ and $G'[V']$, respectively. Thus, any two vertices of $G[V]$ or of $G'[V']$ are adjacent if and only if they alternate in the word $u$.
	
	For $a \in V$ and $a' \in V'$, note that $a$ and $a'$ are adjacent in $G \srec G'$ if and only if $a \in N_G(m)$ and $a' \in N_{G'}(m')$. However, note that $V = N_{G}(m)$ as $m$ is an all-adjacent vertex of $G$.  Accordingly, we consider the two cases based on $a' \in N_{G'}(m')$ or not to show that $a$ and $a'$ are adjacent in $G \srec G'$ if and only if they alternate in $u$ so that $u$ represents the graph $G \srec G'$. 
	
	\begin{itemize}
		\item Case-1: $a' \in N_{G'}(m')$. Since $w'$ represents $G'$, we have $a'm' \preceq p'_i$ for all $i$, or $m'a' \preceq p'_i$ for all $i$, $1 \le i \le t$. Observe that each $u_i$ ($1 \le i \le t$) is constructed from $p'_i$ by replacing $m'$ with the permutation $p_{i_V}$. Hence, if $a'm' \preceq p'_i$ for all $i$, then $a'a \preceq u_i$ for all $i$. If $m'a' \preceq p'_i$ for all $i$, then $aa' \preceq u_i$ for all $i$. In any case, $a$ and $a'$ alternate in $u$.
		
		\item Case-2: $a' \notin N_{G'}(m')$. Then $a'$ and $m'$ do not alternate in $w'$. Accordingly, there exist $i$ and $j$ $(1 \le i, j \le t)$ such that $a'm' \preceq p'_i$ and $m'a' \preceq p'_j$. Again, as per the construction of $u_i$ and $u_j$, we have $a'a \preceq u_i$ and $aa' \preceq u_j$. Hence, $a$ and $a'$ do not alternate in $u$.
	\end{itemize}
	
	Similarly, if $m'$ is an all-adjacent vertex of $G'$, we can show that the word $u = u_1u_2 \cdots u_t$ represents the graph $G \srec G'$ permutationally, where $u_i = r_ir'_is'_is_i$, for $1 \le i \le t$.
 
	Therefore, we have $\mathcal{R}^p(G \srec G') \le t = \max\{k, k'\}$. But, by Remark \ref{com_reco}, since $G$ and $G'$ are isomorphic to certain induced subgraphs of $G \srec G'$, we have $\mathcal{R}^p(G \srec G') \geq t$. Hence, $\mathcal{R}^p(G \srec G') = \max\{k, k'\}$. \qed	 
  
\end{proof}

We now recall a theorem that characterizes a source of a comparability graph using a graph that is constructed as follows. Let $v$ be a vertex of a graph $C$ and $C^v$ be the graph obtained from $C$ by adding a new vertex $v^\prime$ adjacent only to the vertex $v$ of $C$. Then we have the following theorem.

\begin{theorem}[\cite{gimbel1992sources,olariu1992sources}] \label{source}
	Let $C$ be a comparability graph and $v$ be a vertex of $C$. Then $v$ is a source with respect to some transitive orientation of $C$ if and only if $C^v$ is a comparability graph.
\end{theorem}

Now, using Theorem \ref{source}, we find a necessary and sufficient condition to make $G \srec G'$ a comparability graph, when one of $m$ and $m'$ is not a source (or sink) with respect to any transitive orientation of $G$ or $G'$.

\begin{theorem}
Let $G$ and $G'$ be two comparability graphs such that $m$ is not a source with respect to any transitive orientation of $G$. If $G \srec G'$ is a comparability graph, then $m'$ is an all-adjacent vertex of $G'$.	
\end{theorem}

\begin{proof}
	Let $G \srec G'$ be a comparability graph. Assume $m'$ is not an all-adjacent vertex of $G'$. Then there exists a vertex $c' \in V'$ such that $c' \notin N_{G'}(m')$ but $\overline{b'c'} \in E'$ for some vertex $b' \in N_{G'}(m')$. The connectedness of the graph $G'$ ensures the existence of such vertex $c'$. 
	
	Consider the induced subgraph $(G \srec G')[V \cup \{b'\}]$, say ${^m}G_{b'}$. Note that $b'$ is adjacent to all vertices of $N_{G}(m)$ in $G \srec G'$. Hence, ${^m}G_{b'}$ is isomorphic to $G$ as it can be obtained from $G$ by replacing $m$ with $b'$. While ${^m}G_{b'}$ is a comparability graph, $b'$ cannot be a source with respect to any transitive orientation of ${^m}G_{b'}$, as $G$ does. 
	
	Now consider the induced subgraph $(G \srec G')[V \cup \{b', c'\}]$ of $G \srec G'$. Note that $(G \srec G')[V \cup \{b', c'\}] = ({^m}G_{b'})^{b'}$, the graph obtained from ${^m}G_{b'}$ by making $c'$ adjacent only to $b'$. Since $b'$ cannot be a source, $({^m}G_{b'})^{b'}$ is not a comparability graph, by Theorem \refeq{source}. However, $({^m}G_{b'})^{b'}$ is a comparability graph as it is an induced subgraph of the comparability graph $G \srec G'$; a contradiction.  Hence, $m'$ should be an all-adjacent vertex of $G'$. 	\qed
\end{proof}

We summarize the results of this section in the following theorem.

\begin{theorem}\label{main_charc}
	Let $G = (V \cup \{m\}, E)$ and $G' = (V' \cup \{m'\}, E')$ be two comparability graphs with disjoint vertex sets. Suppose $\mathcal{R}^p(G) = k$ and $\mathcal{R}^p(G') = k'$. Then we have the following:
	
	\begin{enumerate}
		\item The split recomposition $G \srec G'$  is a comparability graph if and only if one of the following holds:
		\begin{itemize}
			\item The vertices $m$ and $m'$ are source and sink, or vice versa, with respect to some transitive orientations on $G$ and $G'$.
			\item At least one of the vertices $m$ and $m'$ is an all-adjacent vertex.  
		\end{itemize}
		
		\item Moreover, if the split recomposition $G \srec G'$ is a comparability graph, then $\mathcal{R}^p(G \srec G') = \max\{k, k'\}$ or $1 + \max\{k, k'\}$.
	\end{enumerate} 
\end{theorem} 

As a consequence, we also provide the \textit{prn} of the graph construction $G^m$ in the following corollary.

\begin{corollary}
	Let $G$ be a comparability graph such that $\mathcal{R}^p(G) = k$. If $G^m$ is a comparability graph, then $\mathcal{R}^p(G^m) = k$ or $1 + k$.
\end{corollary}

\begin{proof}
	 If $G^m$ is a comparability graph, then by Theorem \ref{source}, $m$ is a source of $G$. Consider a path $P_3$: $\begin{tikzpicture}[scale=0.6]
	 	\vertex (a_1) at (0,0) [fill=black, label=above:$m'$] {};
	 	\vertex (a_2) at (1,0) [fill=black, label=above:$b$] {};
	 	\vertex (a_3) at (2,0) [fill=black, label=above:$a$] {};
	 	
	 	\path
	 	(a_1) edge (a_2)
	 	(a_2) edge (a_3);
	 \end{tikzpicture}$ and note that $\mathcal{R}^p(P_3) = 2$, as the word $bam'bm'a$ represents $P_3$, which is not a complete graph.  Further, it can be observed that $G^m$ is isomorphic to $G \srec P_3$. 
	
	If $k = 1$, then $G$ is a complete graph so that $m$ is an all-adjacent vertex of $G$. Accordingly, by Theorem \ref{th_all_adjacent}, $\mathcal{R}^p(G^m) = \mathcal{R}^p(G \srec P_3) = \max\{1, 2\} = 2$. Otherwise, for $k \ge 2$, $\max\{k, 2\} = k$. Then, by Theorem \ref{main_charc}, $\mathcal{R}^p(G^m) = k$ or $1 + k$. Hence, in any case, we have $\mathcal{R}^p(G^m) = k$ or $1 + k$. \qed
\end{proof}

\section{\textit{prn}-Irreducible Graphs}

Following the notion of irreducible posets in order theory (cf. \cite{Trotterbook}), in this section, we define the concept of  \textit{prn}-irreducible comparability graphs and study them with respect to split recomposition. In this connection, we characterize \textit{prn}-irreducible graphs such that their recomposition is a comparability graph. Accordingly, we determine the \textit{prn} of the resultant comparability graph based on the \textit{prn}s of its split components. We also observe that the split recomposition of \textit{prn}-irreducible graphs is not \textit{prn}-irreducible.  

\begin{definition}
A comparability graph $C = (V, E)$ is said to be \textit{$k$-prn-irreducible} for some $k \geq 2$, if $\mathcal{R}^p(C) = k$ and $\mathcal{R}^p(C[V \setminus \{a\}]) = k - 1$ for every $a \in V$. Further, a comparability graph is said to be \textit{prn-irreducible} if it is $k$-prn-irreducible for some $k \geq 2$. 
\end{definition}

\begin{remark}
A comparability graph is $2$-\textit{prn}-irreducible if and only if it is an edgeless graph on two vertices. For instance, any connected graph $C$ on at least 3 vertices with $\mathcal{R}^p(C) = 2$ has a vertex $a$ that is not adjacent to all other vertices. If we remove any vertex adjacent to $a$ from $C$, the \textit{prn} of the resultant graph is not one, as it is not a complete graph.    
\end{remark}

\begin{example}
	For $n \ge 3$, consider the crown graph $H_{n,n}$ (a graph which is obtained by deleting a perfect matching from the complete bipartite graph $K_{n, n}$). From \cite{MR2914710}, recall that $\mathcal{R}^p(H_{n,n}) = n$. Let $H'$ be the bipartite graph obtained by deleting a vertex from $H_{n,n}$. Note that $\mathcal{R}^p(H') \le n-1$ as $n-1$ is the size of one of the partite sets of $H'$, and $n-1 \le \mathcal{R}^p(H')$ as $H_{n-1, n-1}$ is an induced subgraph of $H'$ (cf. \cite{khyodeno2}). Hence, $\mathcal{R}^p(H') = n-1$ so that $H_{n, n}$ is $n$-\textit{prn}-irreducible.
\end{example}

\begin{example}
   For $n \ge 3$, the cycle $C_{2n}$ is $3$-\textit{prn}-irreducible. Note that by deleting a vertex from $C_{2n}$ we get the path $P_{2n-1}$ with $2n-1$ vertices. For $n \ge 3$, $\mathcal{R}^p(C_{2n}) = 3$ and $\mathcal{R}^p(P_n) = 2$ (cf. \cite{khyodeno1}). 
\end{example}

Trotter and Moore, as well as Kelly independently and simultaneously obtained the complete list of 3-irreducible posets \cite{kelly77,trotter76}. A complete list of $3$-\textit{prn}-irreducible comparability graphs can be derived from the corresponding posets.  

\begin{proposition}\label{no_all_adja_vertex}
	For $k \ge 2$, any $k$-\textit{prn}-irreducible comparability graph does not have an all-adjacent vertex.
\end{proposition}

\begin{proof}
	On the contrary, suppose $a$ is an all-adjacent vertex of a $k$-\textit{prn}-irreducible comparability graph $C = (V, E)$. As $\mathcal{R}^p(C[V \setminus \{a\}]) = k-1$, let $p_1p_2 \cdots p_{k-1}$ represents the induced subgraph  $C[V \setminus \{a\}]$, where each $p_i$ $(1 \le i \le k-1)$ is a permutation on $V \setminus \{a\}$. Then consider the permutations $q_i = ap_i$, for all $1 \le i \le k-1$, on $V$  and note that the word $q_1q_2 \cdots q_{k-1}$ represents the graph $C$. Hence, $\mathcal{R}^p(C) = k-1$; a contradiction. \qed
\end{proof}

In view of Proposition \ref{no_all_adja_vertex}, we have the following corollary of Theorem \ref{main_charc}.

\begin{corollary}\label{Chac_irr_graphs}
	Let $G$ and $G'$ be comparability graphs such that $G$ is $k$-\textit{prn}-irreducible and $G'$ is $k'$-\textit{prn}-irreducible. Then $G \srec G'$ is a comparability graph if and only if $m$ is a source of $G$ and $m'$ is a sink of $G'$ with respect to some transitive orientations of $G$ and $G'$.
\end{corollary}

We now obtain the permutation-representation number of a split recomposition of two \textit{prn}-irreducible comparability graphs in the following theorem. 

\begin{theorem} \label{theorem_8}
	For $k, k' \geq 3$, let $G$ and $G'$ be $k$-\textit{prn}-irreducible and $k'$-\textit{prn}-irreducible comparability graphs, respectively.  If $G \srec G'$ is a comparability graph, then  $\mathcal{R}^p(G \srec G') = \max\{k, k'\}$.	
\end{theorem}

\begin{proof}
	Suppose $G \srec G'$ is a comparability graph. In view of Corollary \ref{Chac_irr_graphs}, let $m$ be a source of $G$ and $m'$ be a sink of $G'$ with respect to some transitive orientations of $G$ and $G'$. Note that $\mathcal{R}^p(G) = k$, $\mathcal{R}^p(G') = k'$ and $\mathcal{R}^p(G[V]) = k - 1$, $\mathcal{R}^p(G'[V']) = k' - 1$. Along the lines of Lemma \ref{const_perm_GG'}, consider the permutations $p_i$ $(1 \le i \le k-1)$ on $V$, $p'_i$ $(1 \le i \le k'-1)$ on $V'$, $q_i$ $(1 \le i \le k)$ on the vertices of $G$, and $q'_i$ $(1 \le i \le k')$ on the vertices of $G'$ such that the words $w = q_1q_2 \cdots q_{k-1}q_{k}$ and $w' = q'_1q'_2 \cdots q'_{k'-1}q'_{k'}$ represent the graphs $G$ and $G'$, respectively, where
	\[\begin{array}{ll}
		q_{i} = mp_{i}, \quad(1 \le i \le k-1); \quad \quad & q'_{i} = p'_{i}m', \quad(1 \le i \le k'-1); \\
		q_{k} = p_{{k-1}_{V \setminus N_{G}(m)}}m p_{{k-1}_{N_{G}(m)}}; \quad \quad & q'_{k'} = p'_{{k'-1}_{N_{G'}(m')}}m'p'_{{k'-1}_{V' \setminus  N_{G'}(m')}}. \\
	\end{array}\]
	
	Suppose $k' \le k$.  We construct $k$ number of permutations $v_i$ ($1 \le i \le k$) on the vertices of $G \srec G'$ as per the following and show that the word $v = v_1v_2 \cdots v_{k}$ represents the graph $G \srec G'$.
	\begin{align*}
		v_{i} & = p'_ip_i, \quad(1 \le i \le k'-1) \\
		v_{j} & = p'_{k'-1}p_{j}, \quad(k' \le j \le k-1) \\
		v_{k } & = p_{{k-1}_{V \setminus N_{G}(m)}}p'_{{k'-1}_{N_{G'}(m')}}p_{{k-1}_{N_{G}(m)}}p'_{{k'-1}_{V' \setminus N_{G'}(m')}}
	\end{align*} 
	Note that $G[V]$ and $G'[V']$ are induced subgraphs of $G \srec G'$ represented by 
	\begin{align*}
		w_{V} &=  p_1p_2 \; \cdots \; p_{k-1} p_{{k-1}_{V \setminus N_{G}(m)}}p_{{k-1}_{N_{G}(m)}}, \; \text{and}\\
		w'_{V'} &= p'_1p'_2 \; \cdots \; p'_{k'-1}p'_{{k'-1}_{N_{G'}(m')}}p'_{{k'-1}_{V' \setminus  N_{G'}(m')}},
	\end{align*}
	respectively. Further, note that 
	\begin{align*}
		v_{V} &= p_1p_2 \; \cdots \; p_{k-1}p_{{k-1}_{V \setminus N_{G}(m)}}p_{{k-1}_{N_{G}(m)}}\; \text{and}\\
		v_{V'} &= p'_1p'_2 \; \cdots \; p'_{k'-2} \overbrace{p'_{k'-1} \; \cdots \; p'_{k'-1}}^{k - k'+1}p'_{{k'-1}_{N_{G'}(m')}}p'_{{k'-1}_{V' \setminus  N_{G'}(m')}}.
	\end{align*}
	Thus, any two vertices of $G[V]$ or of $G'[V']$ are adjacent if and only if they alternate in the word $v$.
	 
	 For $a \in V$ and $a' \in V'$, note that $a$ and $a'$ are adjacent in $G \srec G'$ if and only if $a \in N_G(m)$ and $a' \in N_{G'}(m')$. Accordingly, we consider the following cases to show that $a$ and $a'$ are adjacent in $G \srec G'$ if and only if they alternate in $v$. If $a \in N_G(m)$ and $a' \in N_{G'}(m')$, then from the construction of $v$, we have $a'a \preceq v_i$ for all $1 \le i \le k$, so that $a$ and $a'$ alternate in $v$. If $a \notin N_{G}(m)$, then for any $a' \in V'$, we have $a'a \preceq v_1$ but $aa' \preceq v_{k}$.  Similarly, if $a' \notin N_{G'}(m')$, then for any vertex $a \in V$, we have $a'a \preceq v_1$ but $aa' \preceq v_{k}$. Hence, $a$ and $a'$ do not alternate in $v$, if $a \notin N_G(m)$ or $a' \notin N_{G'}(m')$. Hence the word $v$ represents $G \srec G'$ permutationally.

	Similarly, when $k \le k'$, we can show that the word $v = v_1v_2 \cdots v_{k'-1}v_{k'}$ represents the graph $G \srec G'$ permutationally, where 
	\begin{align*}
		v_{i} & = p'_ip_i, \quad(1 \le i \le k-1); \\
		v_{j} & = p'_jp_{k-1}, \quad(k \le j \le k'-1); \\
		v_{k'} & = p_{{k-1}_{V \setminus N_{G}(m)}}p'_{{k'-1}_{N_{G'}(m')}}p_{{k-1}_{N_{G}(m)}}p'_{{k'-1}_{V' \setminus N_{G'}(m')}}.
	\end{align*}
	
	In any case, since $v$ is a concatenation of $\max\{k, k'\}$ number of permutations, we have $\mathcal{R}^p(G \srec G') \le \max\{k, k'\}$. Also, since $G$ and $G'$ are isomorphic to certain induced subgraphs of $G \srec G'$ (cf. Remark \ref{com_reco}), we have  $\mathcal{R}^p(G \srec G') \ge \max\{k, k'\}$. Hence,
	$\mathcal{R}^p(G \srec G') = \max\{k, k'\}$. \qed	 
\end{proof}

In view of Theorem \ref{theorem_8}, we have the \textit{prn} of following bipartite graphs which are obtained from the recomposition of \textit{prn}-irreducible graphs.

\begin{example}
	For $k, k' \ge 3$,  the \textit{prn} of the recomposition of crown graphs $H_{k, k}$ and $H_{k', k'}$ is $\max\{k, k'\}$.
\end{example}

\begin{example}
	For $k, k' \ge 3$,  the \textit{prn} of the recomposition of even cycles $C_{2k}$ and $C_{2k'}$ is 3.
\end{example}

\begin{theorem}
  For $k, k' \geq 3$, if $G$ and $G'$ be $k$-\textit{prn}-irreducible and $k'$-\textit{prn}-irreducible comparability graphs, respectively, then $G \srec G'$ is not \textit{prn}-irreducible.
\end{theorem}

\begin{proof}
	Suppose $G \srec G'$ is a comparability graph. By Proposition \ref{no_all_adja_vertex}, $m$ and $m'$ are not all-adjacent vertices. Also, by Theorem \ref{theorem_8}, we have $\mathcal{R}^p(G \srec G') = \max\{k, k'\} = k$, say, without loss of generality.
	
	Since $m'$ is not an all-adjacent vertex of $G'$, there exists $c' \in V'$ such that $c' \notin N_{G'}(m')$. Consider the induced subgraph $L = (G \srec G')[V \cup (V' \setminus \{c'\})]$. 
	Clearly, $\mathcal{R}^p(L) \le k$. We show that $\mathcal{R}^p(L) = k$ so that $G \srec G'$ is not \textit{prn}-irreducible.
	 
	For $b' \in N_{G'}(m')$, note that the induced subgraph $(G \srec G')[V \cup \{b'\}]$ is isomorphic to $G$ and it is also an induced subgraph of $L$. Thus, $\mathcal{R}^p(L) \ge k$ and hence, $\mathcal{R}^p(L) = k$.  \qed
\end{proof}

\section{Conclusion}

This work paves a way, in general, to study the word-representability of graphs using split decomposition and recomposition. As the problem of deciding whether a given graph is word-representable is NP-complete, using split decomposition and recomposition, the problem can be reduced to its prime components. For example, the parity graphs are established as word-representable, because their prime components are word-representable (see Corollary \ref{parity_graphs}). In view of Corollary \ref{perfectgraphs}, the open problem of characterizing word-representable perfect graphs can be approached through the following problem: What are the prime perfect graphs that are not word-representable? Similar to bipartite graphs, it is also interesting to find the subclasses of comparability graphs which are closed under recomposition. If one of the marked vertices is an all-adjacent vertex in a general case, or in the case $G$ and $G'$ are \textit{prn}-irreducible graphs, we proved that $\mathcal{R}^p(G \srec G') = \kappa$, where  $\kappa = \max\{\mathcal{R}^p(G), \mathcal{R}^p(G')\}$.  However, there are two categories of \textit{prn} for the recomposed comparability graph, i.e., $\mathcal{R}^p(G \srec G') = \kappa$ or $1 + \kappa$. Accordingly, one may target to classify the graphs with respect to the \textit{prn} of their recomposition amongst the two categories.

\end{document}